\documentclass[11pt,letterpaper]{article}

\usepackage[utf8]{inputenc}
\usepackage[T1]{fontenc}
\usepackage{lmodern}
\usepackage[DIV=11]{typearea} % large value = less margin (roughly speaking), recommended values: 10pt: 8, 11pt: 10, 12pt: 12

\usepackage{microtype}

\usepackage{amssymb}
\usepackage{amsmath}
\usepackage{amsthm}
\usepackage{thmtools}
\usepackage{thm-restate}

\usepackage[procnumbered,ruled,vlined,linesnumbered]{algorithm2e}

\usepackage{xcolor}
\usepackage{xspace}
\usepackage{xfrac}

\usepackage{comment}

\usepackage[
	style=alphabetic,
	backref=true,
	doi=false,
	url=false,
	maxcitenames=3,
	mincitenames=3,
	maxbibnames=10,
	minbibnames=10,
	backend=bibtex8,
	sortlocale=en_US
]{biblatex}

\usepackage[ocgcolorlinks]{hyperref} % Option ocgcolorlinks makes links only colored on screen and not on printouts
\usepackage{cleveref}

\allowdisplaybreaks[1]

\makeatletter
\g@addto@macro\bfseries{\boldmath}
\makeatother

\makeatletter
\g@addto@macro\mdseries{\unboldmath}
\g@addto@macro\normalfont{\unboldmath}
\g@addto@macro\rmfamily{\unboldmath}
\g@addto@macro\upshape{\unboldmath}
\makeatother

\DeclareCiteCommand{\citem}
    {}
    {\mkbibbrackets{\bibhyperref{\usebibmacro{postnote}}}}
    {\multicitedelim}
    {}

\renewcommand*{\multicitedelim}{\addcomma\space}

\AtEveryCitekey{\clearfield{note}}

\newcommand{\myhref}[1]{%
  \iffieldundef{doi}
    {\iffieldundef{url}
       {#1}
       {\href{\strfield{url}}{#1}}}
    {\href{http://dx.doi.org/\strfield{doi}}{#1}}%
}

\DeclareFieldFormat{title}{\myhref{\mkbibemph{#1}}}
\DeclareFieldFormat
  [article,inbook,incollection,inproceedings,patent,thesis,unpublished]
  {title}{\myhref{\mkbibquote{#1\isdot}}}

\makeatletter
\AtBeginDocument{%
    \newlength{\temp@x}%
    \newlength{\temp@y}%
    \newlength{\temp@w}%
    \newlength{\temp@h}%
    \def\my@coords#1#2#3#4{%
      \setlength{\temp@x}{#1}%
      \setlength{\temp@y}{#2}%
      \setlength{\temp@w}{#3}%
      \setlength{\temp@h}{#4}%
      \adjustlengths{}%
      \my@pdfliteral{\strip@pt\temp@x\space\strip@pt\temp@y\space\strip@pt\temp@w\space\strip@pt\temp@h\space re}}%
    \ifpdf
      \typeout{In PDF mode}%
      \def\my@pdfliteral#1{\pdfliteral page{#1}}% I don't know why % this command...
      \def\adjustlengths{}%
    \fi
    \ifxetex
      \def\my@pdfliteral #1{}% isn't equivalent to this one
      \def\adjustlengths{\setlength{\temp@h}{-\temp@h}\addtolength{\temp@y}{1in}\addtolength{\temp@x}{-1in}}%
    \fi%
    \def\Hy@colorlink#1{%
      \begingroup
        \ifHy@ocgcolorlinks
          \def\Hy@ocgcolor{#1}%
          \my@pdfliteral{q}%
          \my@pdfliteral{7 Tr}% Set text mode to clipping-only
        \else
          \HyColor@UseColor#1%
        \fi
    }%
    \def\Hy@endcolorlink{%
      \ifHy@ocgcolorlinks%
        \my@pdfliteral{/OC/OCPrint BDC}%
        \my@coords{0pt}{0pt}{\pdfpagewidth}{\pdfpageheight}%
        \my@pdfliteral{F}% Fill clipping path (the url's text) with
        \my@pdfliteral{EMC/OC/OCView BDC}%
        \begingroup%
          \expandafter\HyColor@UseColor\Hy@ocgcolor%
          \my@coords{0pt}{0pt}{\pdfpagewidth}{\pdfpageheight}%
          \my@pdfliteral{F}% Fill clipping path (the url's text)
        \endgroup%
        \my@pdfliteral{EMC}%
        \my@pdfliteral{0 Tr}% Reset text to normal mode
        \my@pdfliteral{Q}%
      \fi
      \endgroup
    }%
}
\makeatother

\colorlet{DarkRed}{red!50!black}
\colorlet{DarkGreen}{green!50!black}
\colorlet{DarkBlue}{blue!50!black}

\hypersetup{
	linkcolor = DarkRed,
	citecolor = DarkGreen,
	urlcolor = DarkBlue,
	bookmarks = true,
	bookmarksnumbered = true,
	linktocpage = true
}

\declaretheorem[numberwithin=section]{theorem}
\declaretheorem[numberlike=theorem]{lemma}

\declaretheorem[numberlike=theorem]{claim}

\DontPrintSemicolon
\SetKw{KwAnd}{and}
\SetProcNameSty{textsc}
\SetFuncSty{textsc}

\newcommand{\dist}{d}
\DeclareMathOperator*{\argmax}{arg\,max}

\title{Fully dynamic all-pairs shortest paths with worst-case update-time revisited\thanks{To be presented at the Symposium on Discrete Algorithms (SODA) 2017. This work was done in part while the authors were at Microsoft Research Silicon Valley Lab, Mountain View, USA.}}
\author{
Ittai Abraham\thanks{Hebrew University of Jerusalem, Israel.}
\and
Shiri Chechik\thanks{Tel Aviv University, Israel. This research was supported by the ISRAEL SCIENCE FOUNDATION (grant No. 1528/15).}
\and
Sebastian Krinninger\thanks{Max Planck Institute for Informatics, Saarland Informatics Campus, Germany. Work done in part while at the University of Vienna, Faculty of Computer Science, Austria, and while at the Simons Institute for the Theory of Computing, Berkeley, USA.}
}

\date{}

\hypersetup{
	pdftitle = {Fully dynamic all-pairs shortest paths with worst-case update-time revisited},
	pdfauthor = {Ittai Abraham, Shiri Chechik, Sebastian Krinninger}
}

\begin{document}
\maketitle
\begin{abstract}
We revisit the classic problem of dynamically maintaining shortest paths between all pairs of nodes of a directed weighted graph.
The allowed updates are insertions and deletions of nodes and their incident edges.
We give worst-case guarantees on the time needed to process a single update (in contrast to related results, the update time is \emph{not} amortized over a sequence of updates).

Our main result is a simple randomized algorithm that for any parameter $c>1$ has a worst-case update time of $ O (cn^{2+\sfrac{2}{3}} \log^{\sfrac{4}{3}}{n}) $ and answers distance queries correctly with probability $1-1/n^c$, against an adaptive online adversary if the graph contains no negative cycle.
The best deterministic algorithm is by Thorup \citem[STOC 2005]{Thorup05} with a worst-case update time of $ \tilde O (n^{2+\sfrac{3}{4}})$ and assumes non-negative weights. This is the first improvement for this problem for more than a decade. 
Conceptually, our algorithm shows that randomization along with a more direct approach can provide better bounds. 

\end{abstract}
\newpage

\section{Introduction}

In the all-pairs shortest paths (APSP) problem we are interested in computing the distance matrix of a given graph.
In the \emph{fully dynamic} version of this problem the graph might undergo updates in the form of insertions and deletions of nodes and their incident edges.
The goal is to refresh the distance matrix after each such update as quickly as possible.
In particular we want algorithms that are more efficient than recomputing the distance matrix from scratch after every update in the graph.
The time needed to perform the operations for refreshing the matrix is called update time.
Our main result is a fully dynamic APSP algorithm for weighted directed graphs with $ n $ nodes. Our algorithm is randomized and answers queries correctly with \textit{high probability} (the probability of error is polynomially small ($n^{-c}$)) against an \textit{adaptive online adversary}~\cite{Ben-DavidBKTW94}.
This type of adversary cannot see the algorithm's random coins and internal data structure, but it may choose arbitrary updates and path queries based on all previous responses of the algorithm.
\begin{theorem}\label{thm:worst_case_randomized}
For every $ c \geq 1 $, there is a randomized fully dynamic algorithm for maintaining the distance matrix of a weighted directed graph containing no negative cycles that, with probability at least $ 1 - 1 / n^c $, has a worst-case update time of $ O (c n^{2 + \sfrac{2}{3}} \log^{\sfrac{4}{3}} {n}) $ and is correct against an adaptive online adversary.
\end{theorem}

In the past years (see \Cref{sec:related_work}) there has been significant progress on dynamic shortest path problems. However, obtaining worst-case bounds for the most general setting of arbitrary graphs with weights seems challenging. Despite considerable recent attention to dynamic graph problems, the only result in this model is the deterministic algorithm of
Thorup~\cite{Thorup05} from STOC 2005 that obtained worst-case update time of $ \tilde O (n^{2+\sfrac{3}{4}}) $ for non-negative weights.\footnote{We use $\tilde{O}(\cdot)$ to hide polylogarithmic factors in $ n $.}

We present the first solution that takes advantage of randomization to improve the worst-case update time from $ \tilde O (n^{2+\sfrac{3}{4}}) $ to $ \tilde O (n^{2+\sfrac{2}{3}}) $ and extend it to negative weights. We believe this trade-off (significantly better update time at the cost of negligible  ($n^{-c}$) probability of being wrong) is an important point on the design space of fully dynamic shortest path algorithms.
In addition, our solution is arguably simpler than the one of Thorup~\cite{Thorup05}.
The algorithm of Thorup~\cite{Thorup05} relies on the algorithm of Demetrescu and Italiano~\cite{DemetrescuI04} and it is essentially
a sophisticated de-amortization of Demetrescu and Italiano~\cite{DemetrescuI04}.
In contrast, both our algorithm and analysis are pretty simple and independent of other sophisticated algorithms. 
Although our algorithm is not deterministic, its guarantees are stronger than those of many other randomized shortest paths algorithms.
Namely, it supports updates by an adaptive online adversary and not just an oblivious adversary.
This means that the adversary is allowed to base its updates on to the shortest paths previously returned by the algorithm (but does not directly see the algorithm's internal randomness).
The weaker oblivious adversary has to fix its sequence of updates before the algorithm starts.

Our algorithms compute the distance matrix explicitly after every update.
In general, this is not required for a fully dynamic algorithm as long as it is able to answer queries for the distances between nodes after each update step.
The time needed to perform a single query is called query time.
Our algorithms have constant query time when asked for the distance between two nodes.
They can easily be extended to also output the shortest path connecting two nodes in time proportional to the length of the path at the cost of an additional factor of $ \log{n} $ in the update time (see \Cref{sec:returning shortest paths}).
Restricting the allowed updates to insertions and deletions of nodes is no loss of generality:
insertions and deletions of edges as well as edge weight increases can be simulated by at most two node updates.

\subsection{Additional results}

We believe that Sankowski's framework for maintaining the matrix adjoint~\cite{Sankowski04,Sankowski05} gives a randomized fully dynamic algorithm for maintaining the distance matrix of an \emph{unweighted} directed graph with a worst-case update time of $ \tilde O (n^{2+{1/2}}) $. However it seems that this extension is inherently limited to maintaining distances and cannot efficiently be extended to output also the shortest path connecting two nodes in time proportional to the length of the path.

In Section \ref{appendix:unweighted} we resolve this shortcoming. We show how to extend our scheme to \emph{unweighted} graphs with a worst-case update time of $ \tilde O (n^{2+{1/2}}) $ and allow to output also the shortest path connecting two nodes in time proportional to the length of the path (see also \Cref{sec:returning shortest paths}).

Finally in Section \ref{appendix:deterministic} we show that our scheme can be extended to a \emph{deterministic} version with \emph{negative weights} and obtain a worst-case update time of $ O (n^{2+\sfrac{3}{4}} \log^{2/3} n)$. This result improves on the best known deterministic result by reducing the logarithmic factors. % allowing 

\subsection{Recomputing from scratch}

An alternative approach is to recompute all-pairs shortest paths from scratch on each update. Fully dynamic algorithms (like ours and Thorup's) improve on this approach when the edge weights can be relatively large  or when the graph is unweighted (like ours and Sankowski's).
In the static setting, Zwick's pseudopolynomial all-pairs shortest paths algorithm~\cite{Zwick02} has a running time of $ O (n^{2.5302}) $~\cite{Gall12} if its input is a directed graph with integer edge weights from $ \{ -W, \ldots, 0, \ldots, W \} $ such that $ W \leq n^{3-\omega} $.
However, Zwick's algorithm achieves its superior running time in this regime by using a fast rectangular matrix multiplication algorithm as a subroutine.
Large constants in the running times of these algorithms make it worthwhile to find solutions that do not rely on fast matrix multiplication as a subroutine.
For arbitrary edge weights, the current fastest algorithm has a running time of $ O (n^3 / 2^{\Omega(\log^{\sfrac{1}{2}}{n})}) $~\cite{Williams14,ChanW16}.

\subsection{Related work}\label{sec:related_work}

Unless noted otherwise, the algorithms cited in this section are deterministic and allow insertions and deletions of nodes and their incident edges in directed graphs with non-negative edge weights.

\paragraph{Fully dynamic algorithms.}
The study of fully dynamic APSP algorithms for general directed graphs was initiated by King~\cite{King99}.
She obtained a fully dynamic algorithm with a pseudopolynomial amortized update time of $ O (n^{2+\sfrac{1}{2}} \sqrt{W \log{n}}) $, where the edge weights have to be positive integers and $ W $ is the largest among the weights.
She also presented $ (1 + \epsilon) $- and $ (2 + \epsilon) $-approximations (for any positive constant $ \epsilon $) with amortized update times of $ O (n^2 \log^3{(W n)} / \epsilon^2) $ and $ O (n^2 \log^2{n} / \log{\log{n}}) $, respectively.

Later, Demetrescu and Italiano~\cite{DemetrescuI06} obtained an algorithm allowing arbitrary edge weight updates with an amortized update time of $ O (n^{2+\sfrac{1}{2}} \sqrt{S \log^3{n}}) $, where each edge can assume at most $ S $ different real values.
Using a new framework for exploiting local properties of shortest paths Demetrescu and Italiano~\cite{DemetrescuI04} obtained an algorithm with an amortized update time of $ O (n^2 \log^3{n}) $.
This result is essentially optimal (up to logarithmic factors) if we demand that the algorithm has to maintain the distance matrix explicitly.
Thorup~\cite{Thorup04} slightly improved this update time to $ O (n^2 (\log{n} + \log^2{((n+m)/n)})) $, thus subsuming previous results on this problem in terms of running time.
Subsequently, Thorup~\cite{Thorup05} developed an algorithm with a worst-case update time of $ \tilde O (n^{2+\sfrac{3}{4}}) $ by deamortizing~\cite{DemetrescuI06}.
There are also two randomized algorithms for unweighted directed graphs with non-constant query time.
The first one by Roditty and Zwick has an amortized update time of $ \tilde O (m \sqrt{n}) $ and a query time of $ O (n^{\sfrac{3}{4}}) $.
The second one by Sankowski~\cite{Sankowski05} uses fast matrix multiplication as a subroutine and has a worst-case update time of $ O (n^{1.932}) $ and a query time of $ O (n^{1.288}) $.

Further results include fully dynamic algorithms for planar~\cite{KleinS98,HenzingerKRS97,FakcharoenpholR06,AbrahamCG12} and undirected graphs~\cite{RodittyZ12,Bernstein09,Bernstein16,AbrahamCT14}.

\paragraph{Partially dynamic algorithms.}
In the partially dynamic model only one type of updates is allowed.
The \emph{incremental} model is restricted to insertions and the \emph{decremental} model is restricted to deletions.
The amortized update time bounds of partially dynamic algorithms in all known algorithms do not depend on the number of updates performed.
Thus it is often convenient to report the \emph{total update time}, which is the sum of the individual update times.

A simple incremental algorithm for inserting a single node can be obtained by running one iteration of the Floyd-Warshall algorithm, which takes time $ O(n^2) $, or, total update time $ O(n^3) $ for inserting $ n $ nodes.
Ausiello et al.~\cite{AusielloIMN91} presented an incremental algorithm for inserting edges or decreasing edge weights in integer weighted graphs with a total update time of $ O (n^3 W \log{(nW)}) $, where $ W $ is the largest edge weight.

In terms of node deletions, the algorithm of Demetrescu and Italiano~\cite{DemetrescuI04} has a total update time of $ O (n^3 \log{n}) $.
The fastest decremental algorithms for edge deletions have total update times of randomized $ O (n^3 \log^2{n}) $ in unweighted graphs~\cite{BaswanaHS07}, or $ O (n^3 S \log^3{n}) $ in weighted graphs~\cite{DemetrescuI06}, where $ S $ is the number of different values each edge assumes (edge deletions are implemented by setting the weight of the edge to $ \infty $).
If we allow approximate answers, the state of the art is a randomized algorithm by Bernstein~\cite{Bernstein16}: in directed graphs with integer edge weights decremental $ (1 + \epsilon) $-approximate APSP can be maintained with a total update time of $ \tilde O (m n \log{W} / \epsilon) $, where $ W $ is the largest edge weight.\footnote{Equivalently, if the graph has rational edge weights, the total update time is $ \tilde O (m n \log{R} / \epsilon) $, where $ R $ is the ratio of the largest to the smallest edge weight.}
Related work includes various approximation algorithms for undirected graphs~\cite{BaswanaHS03,RodittyZ12,BernsteinR11,HenzingerKNSICOMP16,AbrahamC13,HenzingerKNSODA14,AbrahamCT14}, in particular also for the single-source shortest paths problem~\cite{EvenS81,BernsteinR11,HenzingerKNICALP13,HenzingerKNSODA14,HenzingerKNSTOC14,BernsteinC16}.

\section{Preliminaries}
In the rest of this paper we consider a weighted directed graph $ G $ undergoing insertions and deletions of nodes and their incident edges.
At every insertion of a node its incoming and outgoing edges and their respective weights are specified.
We define $ V $ and $ E $ to be the sets of nodes and edges of $ G $, respectively.
We set $ n = |V| $ and $ m = |E| $.
Given a subset $ S \subseteq V $ of nodes we denote by $ G \setminus S $ the subgraph of $ G $ induced by $ V \setminus S $.
The weight of an edge $ (u, v) \in E $ is denoted by $ w (u, v) $.
We define the length of a path to be the sum of the weight of its edges.
The shortest path from $ s $ to~$ t $ is the minimum length path from $ s $ to~$ t $ (or $ \infty $ if no such path exists).
The distance from $ s $ to~$ t $ in a graph $ G $ is the length of the shortest path from $ s $ to~$ t $ and is denoted by $ \dist_G (s, t) $.
A $ \leq h $ hop path is a path consisting of at most $ h $ edges.
The shortest $ \leq h $ hop path from $ s $ to $ t $ is a path with minimum length among all $ \leq h $ hop paths from $ s $ to $ t $.
We denote by $ \dist_G^h (s, t) $ the length of the shortest $ \leq h $ hop path from $ s $ to $ t $ (or $ \infty $ if no such path exists).
Note that shortest $ \leq h $ hop paths may be different from shortest paths in the case where the shortest paths contain more than $h$ edges.

In our algorithm we use the well-known fact that good hitting sets can be obtained by random sampling.
This technique was first used in the context of shortest paths by Ullman and Yannakakis~\cite{UllmanY91}.
A general lemma can be formulated as follows.

\begin{lemma}\label{lem:random hitting set}
Let $ a \geq 1 $, let $ T $ be a set of size $ t $ and let $ S_1, S_2, \ldots, S_k $ be subsets of $ T $ of size at least $ q $.
Let~$ U $ be a subset of $ T $ that was obtained by choosing each element of $ T $ independently with probability $ p = \min (x / q, 1) $ where $ x = a \ln{(k t)} + 1 $.
Then, with high probability (whp), i.e., probability at least $ 1 - 1/t^a $, the following two properties hold:
\begin{enumerate}
\item For every $ 1 \leq i \leq k $, the set $ S_i $ contains a node of $ U $, i.e., $ S_i \cap U \neq \emptyset $.
\item $ |U| \leq 3 x t / q = O (a t \ln{(k t)} / q) $.
\end{enumerate}
\end{lemma}

A second ingredient of our algorithm is a simple technique for handling adversarial distribution of loads.
The lemma below was observed by Levcopoulos and Overmars~\cite{LevcopoulosO88} in the context of balanced search trees.
Thorup later applied it to bounding the number of precomputed paths through certain nodes in his fully dynamic APSP algorithm~\cite{Thorup05}.

\begin{lemma}[\cite{LevcopoulosO88}]\label{lem:distributing stones on piles}
Consider the following process for repeatedly distributing $ L $ stones on $ k $ piles.
In each round, remove the pile with the maximum number of stones (together with the stones it contains) and let an adversary distribute a total of at most $ L $ stones on the remaining piles.
Then, at any time, the maximum number of stones on any pile is $ O (L \log{k}) $.
\end{lemma}

Our last ingredient is a reduction for obtaining a fully dynamic algorithm from a decremental algorithm. This approach was first taken by Henzinger and King for the dynamic minimum spanning tree (MST) problem \cite{HenzingerK01} and a later variant was used by Thorup~\cite{Thorup05}.
Consider a decremental algorithm that, after a preprocessing stage, can handle a single batch of up to $ 2 \Delta $ deletions. First we reduce the cost of the pre-processing by a factor of $\Delta$ at the cost of supporting at most $\Delta$ deletions (instead of $2\Delta$). This is done by keeping two copies: at every interval of $\Delta$ operations, one copy is used to answer queries and the other copy is being gradually built (in each round a $1/\Delta$ fraction of the pre-computation is executed). When the gradually built copy is ready to serve queries, it is $\Delta$ rounds behind, but this is okay because it can handle $2\Delta$ deletions and we simply add the at most $\Delta$ deletions of the previous interval to the at most $\Delta$ deletions on the current interval.

Second, reducing this decremental-only algorithm to a fully dynamic one at the cost of $O(\Delta n^2)$ time per update is done by running a modification of the Floyd-Warshall algorithm.
Recall that standard Floyd-Warshall algorithm operates in iterations.
In each iteration the algorithm selects a new vertex and updates the shortest paths of all pairs by allowing them to use the new selected node together will all previously selected nodes.
Notice that for our needs, this means that we can start from the result of the decremental algorithm and do only $ 2 \Delta $ iterations, one for each new inserted node.
Each iteration takes time $ O (n^2) $ so the entire process takes time $ O (\Delta n^2) $.
Note that we do this operation from scratch after every update. That is, after every update (deletion or insertion) the algorithms invokes these
$O(\Delta)$ Floyd-Warshall iterations. Therefore, we only need to maintain the decremental data structure dynamically.
In \Cref{sec:negative_weights} we extend this reduction such that the graph may have negative weights, but no negative cycles, while the decremental algorithm only needs to work with non-negative edge weights.

\begin{lemma}[\cite{HenzingerK01},\cite{Thorup05}]\label{lem:decremental_to_fully_dynamic}
If there is a decremental APSP algorithm supporting any sequence of up to $ 2 \Delta $ deletions that spends time $ t_\text{pre} $ for preprocessing the initial graph and worst-case time $ t_\text{del} $ per deletion, then there is also a fully dynamic APSP algorithm with a worst-case update time of $ O(t_\text{pre} / \Delta + t_\text{del} + \Delta n^2) $.
The query time of the fully dynamic algorithm is proportional to the query time of the decremental algorithm.
\end{lemma}

\section{Decremental shortest paths for a batch of deletions}\label{sec:randomized algorithm}

In this section we design a randomized algorithm with the following properties.
We are given a directed graph with non-negative edge weights and preprocess it in time $ \tilde O (n^3) $.
After the preprocessing phase, a batch of nodes $ D $ to be deleted from the graph is given to us and we have to compute the all-pairs distances in $ G \setminus D $.
We will perform this task in time $ \tilde O (n^{2.5} \sqrt{|D|}) $.

\begin{theorem}[Batch deletion algorithm]
Given a graph $ G = (V, E) $ and a parameter $ c \geq 1 $, \Cref{alg:preprocessing} computes a data structure $ \mathcal{D} $ in time $ O (n^3 \log^2{n}) $ such that given $ \mathcal{D} $ and a single set of nodes $ D \subseteq V $, \Cref{alg:deletion} computes the all-pairs distances of $ G \setminus D $ in time $ O (n^{2.5} \sqrt{|D|} \log{n}) $.
The running time bounds and the correctness each hold with probability at least $ 1 - 1 / n^c $.
\end{theorem}

Using the reduction of \Cref{lem:decremental_to_fully_dynamic} this immediately implies our main result (\Cref{thm:worst_case_randomized}) by setting $ \Delta = \lceil n^{1/3} \log^{2/3}{n} \rceil $.
We now describe the decremental algorithm and then analyze its correctness and its running time.

\subsection{Algorithm description}

Our algorithm follows a hierarchical approach where, for every $ 1 \leq i \leq \lceil \log{n} \rceil $, layer $ i $ is used to obtain the shortest paths whose number of edges is between $ h_i/2 $ and $ h_i $ where $ h_i = 2^i $.

\paragraph{Preprocessing.}
For every layer $ i $ (with $ 1 \leq i \leq \lceil \log{n} \rceil $), in the preprocessing phase, we first randomly sample a set $C_i$ of $ \tilde O (n / h_i) $ nodes, called centers, which with high probability will `hit' every shortest path in the graph that has at least $ h_i/2 $ edges.
The bound on the size of $ C_i $ is guaranteed with high probability.
After the sampling we visit a subset $R_i$ of the nodes in a specific order that is determined by the algorithm `on-the-fly'.
Every time we visit a node $ v $, we perform the following operations in the graph $ G_i^v $ from which all edges incident to previously visited nodes have been removed:
for every node $ x $ we compute the shortest $ \leq h_i $ hop path $ \pi_i^v (v, x) $ from $ v $ to $ x $ and the shortest $ \leq h_i $ hop path $ \pi_i^v (x, v) $ from $ x $ to $ v $ by running $ h_i $ iterations of the Bellman-Ford algorithm in $ G_i^v $ and its reverse graph.
To be precise, we compute the shortest $ \leq h_i $ hop paths with the minimum number of edges.
We keep a counter for every node $ u $ to count the `congestion' of $ u $, which for us is the total number of shortest $ \leq h_i $ hop paths containing $ u $ computed in the preprocessing.

The order in which we visit the nodes is determined as follows.
Until all centers have been visited, we alternate between visiting the center with the largest congestion and the non-center node with the largest congestion.
We will show that this greedy strategy limits the maximum congestion of each node, which in turn bounds the work we have to do for updating the shortest paths if this node is deleted later on.
Additionally, we compute for every pair of nodes $ s $ and $ t $ and every visited node $ v $ the distance $\delta_i (s, v, t)$ from $ s $ to $ t $ through $ v $ using the in- and out-trees computed for $ v $; we then sort all nodes $ v \in R_i $ according to their value $\delta_i (s, v, t)$ and store them in a list $ L_i (s, t) $.
\Cref{alg:preprocessing} shows the pseudocode for the preprocessing.
Observe that a shortest path $ \pi $ from $ s $ to $ t $ can by reconstructed from the stored paths if, for the hop range $ h_i / 2 \ldots h_i $ of $ \pi $, we identify the node $ v $ on $ \pi $ that was visited first, which means that $ \pi $ is contained in $ G^v $, and concatenate $ \pi_i^v (s, v) $ and $ \pi_i^v (v, t) $.
Thus, the minimum of all $ \delta_i (s, v, t) $ gives the distance from $ s $ to $ t $ (and we can construct the corresponding shortest path).
Whenever a node $ u $ is deleted from the graph we destroy some of the paths $ \pi_i^v (s, v) $ or $ \pi_i^v (v, t) $ and for such nodes $ s $ and $ t $ we will recompute the shortest paths from and to $ v $.
The number of destroyed paths is equal to the congestion of $ u $, which is our motivation for limiting the maximum congestion of each node.

\begin{algorithm}
\caption{Preprocessing a graph $ G = (V, E) $}
\label{alg:preprocessing}

\SetKwProg{procedure}{Procedure}{}{}
\SetKwFunction{preprocess}{Preprocess}
\SetKwFunction{sample}{Sample}
\SetKwFunction{visit}{Visit}
	
\procedure{\preprocess{$G$}}{
	\For{ $i = 1 $ \KwTo $ \lceil \log{n} \rceil $}{
		$ h_i \gets 2^i $\;
		$ x \gets 1 + 3 (c+1) \ln{n} $\;
		$ C_i \gets \emptyset $	\tcp*{Set of randomly chosen centers}
		\lForEach{$ v \in V $}{
			$ C_i \gets C_i \cup \{ v \} $ with probability $ \min (x/h_i, 1) $
		}
		$ R_i \gets \emptyset $ \tcp*{Set of visited nodes}
		\ForEach{$ v \in V $}{
			$ c_i (v) \gets 0 $ \tcp*{Initialize congestion counter}
		}
		\While(\tcp*[f]{Repeat until all centers have been visited}){$ C_i \setminus R_i \neq \emptyset $}{
			$ v \gets \argmax_{v \in C_i \setminus R_i} c_i (v) $ \tcp*{Visit center with largest congestion}
			\ProcNameSty{Visit}($v$, $i$)\;
			\If(\tcp*[f]{If there are unvisited non-center nodes}){$ V \setminus (C_i \cup R_i) \neq \emptyset $}{
				\tcp{Visit non-center node with largest congestion}
				$ v \gets \argmax_{v \in V \setminus (C_i \cup R_i)} c(v) $\;
				\ProcNameSty{Visit}($v$, $i$)\;
			}
		}
		\ForEach{pair of nodes $ s, t \in V $}{
			Construct list $ L_i (s, t) $ containing all nodes $ v \in R_i $ sorted by their value $ \delta_i (s, v, t) $\;
		}
	}
}

\BlankLine

\procedure{\visit{$v$, $i$}}{
	Construct $ G_i^v = (V, E \setminus (V \times R_i \cup R_i \times V)) $ (in which all edges incident to previously visited nodes are removed)\;
	Run $ h_i $ iterations of Bellman-Ford from $ v $ in $ G_i^v $ and its reverse graph to compute for every node $ x \in V \setminus R_i $ the shortest $ \leq h_i $ hop path $ \pi_i^v (v, x) $ from $ v $ to $ x $ and the shortest $ \leq h_i $ hop path $ \pi_i^v (x, v) $ from $ x $ to $ v $ in $ G_i^v $\;
	$ R_i \gets R_i \cup \{ v \} $\;	
\ForEach{$ u \in V \setminus R_i $}{
		\tcp{Increase congestion counter of $ u $ by number of $ \leq h_i $ hop shortest paths containing $ u $}
		$ c_i (u) \gets c_i (u) + | \{ x \in V \mid u \in \pi_i^v (x, v) \text{ or } \pi_i^v (v, x) \} | $\;
	}
	\ForEach{pair of nodes $ s, t \in V $}{
		$ \delta_i (s, v, t) \gets \dist^{h_i}_{G_i^v} (s, v) + \dist^{h_i}_{G_i^v} (v, t) $\;
	}
}
\end{algorithm}

\paragraph{Batch deletions.}
A batch of deletions given by a set $ D $ is processed by first recomputing the shortest paths consisting of up to $ h = \sqrt{n / |D|} $ edges and then recomputing the shortest paths consisting of more than $ h $ edges.
To find the shortest paths with at most $ h $ edges, we perform the following steps for every $ 1 \leq i \leq \lceil \log{h} \rceil $.
For every $ v \in R_i$, we first iterate over every path $ \pi_i^v (x, v) $ from the preprocessing stage to determine every node $ x $ for which the path \emph{to} $ v $ has been destroyed by one of the deletions in $ D $ and store these nodes in the set $ U_i^v $ of \emph{affected} nodes.
Similarly, we add to $ U_i^v $ all nodes whose shortest path \emph{from} $ v $ has been destroyed.
Note that for all non-affected nodes the shortest $ \leq h_i $ path from and to $v$ did not change since the preprocessing.

We next explain how to compute the new shortest paths from and to $v$ for nodes in $ U_i^v $ as follows.
We compute a sparse sketch graph $ H_i^v $ consisting of the following edges.
For each affected node $ x \in U_i^v $ we add all its incident edges (both incoming and outgoing) to $ H_i^v $.
For each non-affected node $ x \notin U_i^v $ we include two edges: one to the successor of $ x $ on $ \pi_i^v (x, v) $ and one to the predecessor of $ x $ on  $ \pi_i^v (v, x) $.
We will show that by this rule all shortest paths from and to $ v $ in $ G_i^v \setminus D $ with at most $ h_i $ edges are present in the sketch graph.
We then run Dijkstra's algorithm to determine the shortest paths to and from $ v $ in $ H_i^v $.
Furthermore, for each affected node $ x \in U_i^v $ and every possible start or endpoint $ y $, we recompute $ \delta_i (x, v, y) $, the distance from $ x $ to $ y $ through $ v $ in $ H_i^v $ and $ \delta_i (y, v, x) $, the distance from $ y $ to $ x $ through $ v $ in $ H_i^v $, possibly replacing the corresponding value computed in the preprocessing.

To find the shortest paths consisting of more than $ h $ edges, we first run Dijkstra's algorithm to compute the shortest paths to and from every center in $ C_{\lceil \log{h} \rceil} $ and then compute the shortest paths through these centers.
Finally, we determine, for every pair of nodes $s$ and $t$, the shortest path distance $\delta (s, t)$ as the length of the shortest path through any of the centers used in this algorithm.
\Cref{alg:deletion} shows the pseudocode for batch deletions.

\begin{algorithm}
	\caption{Deletion procedure for a single batch of nodes $ D $}
	\label{alg:deletion}

	$ h \gets \sqrt{n / |D|} $\;
	\For{$i = 1 $ \KwTo $ \lfloor \log{h} \rfloor $}{
		\ForEach{$ v \in R_i $}{
			\tcp{Iterate over all paths from preprocessing to determine affected nodes whose shortest $ \leq h_i $ hop paths to or from $ v $ was destroyed}
			$ U_i^v \gets \{ x \in V \mid \pi_i^v (x, v) \cap D \neq \emptyset \text{ or } \pi_i^v (v, x) \cap D \neq \emptyset\} $\;
			\tcp{Construct sketch graph $ H_i^v = (V, E_i^v) $}
			$ E_i^v \gets \emptyset $\;
			\tcp{Add edges to and from all neighbors for affected nodes}
			\lForEach{$ y \in U_i^v $}{	
				$ E_i^v \gets E_i^v \cup \{ (y, z) \mid (y, z) \in E \} \cup \{ (z, y) \mid (z, y) \in E \} $
			}
			\tcp{Add edges of paths from the preprocessing stage for unaffected nodes}
			\ForEach{$ y \notin U_i^v $}{ % $ y \in V \setminus R^v_i $
				Determine predecessor $ x $ of $ y $ on $ \pi_i^v (v, y) $ and successor $ z $ of $ y $ on $ \pi_i^v (y, v) $\;
				$ E_i^v \gets E_i^v \cup \{ (x, y), (y, z) \} $ \;
			}
			Compute SSSP from and to $ v $ in $ H_i^v = (V, E_i^v) $ using Dijkstra's algorithm\;
			\tcp{Update shortest paths through centers for pairs with at least one affected node}
			\lForEach{$ (s, t) \in U_i^v \times V $ and $ (s, t) \in V \times  U_i^v $}{
				$ \delta_i (s, v, t) \gets \dist_{H_i^v} (s, v) + \dist_{H_i^v} (v, t) $
			}

		}
		\ForEach{pair of nodes $ s, t \in V $}{
			$ \delta_i (s, t) \gets \min_{v \in R_i} \delta_i (s,v, t)$ \tcp{see \Cref{lem:time per deletion} for implementation using $ L_i (s, t) $} \label{line:computing_min}
		}
	}
	\lForEach{$ v \in C_{\lceil \log{h} \rceil} $}{
		Compute SSSP from and to $ v $ in $ G \setminus D $ using Dijkstra's algorithm
	}
	\ForEach{pair of nodes $ s, t \in V $}{
			$ \delta_{\lceil \log{h} \rceil} (s, t) \gets \min_{v \in C_{\lceil \log{h} \rceil}} (\dist_{G \setminus D} (s, v) + \dist_{G \setminus D} (v, t)) $\;
			$ \delta (s, t) \gets \min_{1 \leq i \leq \lceil \log{h} \rceil} \delta_i (s, t) $\;
	}
\end{algorithm}

\subsection{Correctness}\label{sec:correctness}

We have to show that our algorithm can serve a batch deletion correctly with probability at least $ 1 - 1 / n^c $.
Just for the following analysis we assume that among all shortest paths for a fixed pair of nodes there is a single \emph{designated} shortest path (e.g., the smallest according to some order).
In order to prove this statement we first show that the sketch graphs used in our algorithm contain the shortest paths relevant to us.

\begin{claim}\label{lem:sketch contains shortest paths with few hops}
Let $ 1 \leq i \leq \lceil \log{n} \rceil $.
For every $ v \in R_i $ and all pairs of nodes $ s, t \in G_i^v \setminus D $,
\begin{itemize}
\item if there is a shortest path from $ s $ to $ v $ in $ G_i^v \setminus D $ with at most $ h_i $ edges, then $ \dist_{H_i^v} (s, v) = \dist_{G_i^v \setminus D} (s, v) $
\item if there is a shortest path from $ v $ to $ t $ in $ G_i^v \setminus D $ with at most $ h_i $ edges, then $ \dist_{H_i^v} (v, t) = \dist_{G_i^v \setminus D} (v, t) $.
\end{itemize}
\end{claim}

\begin{proof}
We only give a proof of the first item as the proof of the second item is symmetric.
Let $ \pi $ denote the designated shortest path from $ s $ to $ v $ in $ G_i^v \setminus D $ with at most $ h_i $ edges.
The proof is by induction on the number of edges of $ \pi $.
Remember that $ \pi_i^v (s, v) $ denotes the shortest $ \leq h_i $ hop path from $ s $ to $ v $ in $ G_i^v $ determined in the preprocessing.
Let $ x $ be the successor of $ s $ on $ \pi $ and let $ y $ be the successor of $ s $ on $ \pi_i^v (s, v) $.

If $ s \in U_i^v $, then the edge $ (s, x) $ is contained in $ H_i^v $ by the definition of $ H_i^v $ and by applying the induction hypothesis on $ x $ we get $ \dist_{H_i^v} (s, v) = \dist_{G_i^v \setminus D} (s, v) $.
If $ s \notin U_i^v $, then $ \pi_i^v (s, v) $ does not contain any deleted nodes and the edge $ (s, y) $ is contained in $ H_i^v $.
The weight of $ \pi_i^v (s, v) $ equal the weight of $ \pi $ because otherwise $ \pi $ would have been a shorter $ \leq h_i $ hop path than $ \pi_i^v (s, v) $ in $ G_i^v $ during the preprocessing.
Thus, $ \pi_i^v (s, v) $ is a shortest path in $ G_i^v \setminus D $.
Furthermore the number of edges of $ \pi_i^v (s, v) $ is at most the number of edges of $ \pi $ as $ \pi_i^v (s, v) $ is the shortest $ \leq h_i $ hop path with the \emph{minimum} number of edges in $ G_i^v $.
Therefore we may apply the induction hypothesis on $ y $ and conclude that $ \dist_{H_i^v} (s, v) = \dist_{G_i^v \setminus D} (s, v) $.
\end{proof}

\begin{claim}\label{lem:centers hit paths}
With probability at least $ 1 - 1/n^c $, for all pairs of nodes $ s, t \in V $, if the designated shortest path from $ s $ to $ t $ has at least $ h_i / 2 $ hops, then it contains a center $ v \in C_i $.
\end{claim}

\begin{proof}
We argue that all lexicographic shortest paths with at least $ 2^{i-1} $ edges contain a node of $ C_i $ with probability at least $ 1 - 1/n^c $.
We apply \Cref{lem:random hitting set} with $ a = c $, $ T = V $, $ U = C_i $, $ t = n $, $ q = h_i/2 $ and by defining $ S_1, \ldots, S_k $ with $ k \leq n^2 $ as the sets of nodes on the at most $ n^2 $ lexicographic shortest paths of pairs of nodes in $ G \setminus D $ with at least $ (h_i / 2) $ edges.
Remember that the algorithm performs the sampling by picking each node with probability $ \min (x/q, 1) $ where $ x \leq c \ln{n^3} + 1 $.
Thus, all lexicographic shortest paths with at least $ 2^{i-1} $ edges contain a center from $ C_i $ with probability at least $ 1 - 1 / n^c $.
\end{proof}

\begin{claim}
With probability at least $ 1 - 1/n^c $,  $ \delta (s, t) = \dist_{G \setminus D} (s, t) $ for all pairs of nodes $ s $ and~$ t $.
\end{claim}

\begin{proof}
We prove the claim by assuming that the statement of \Cref{lem:centers hit paths} holds.
Thus, the claim we prove also holds with probability at least $ 1 - 1/n^c $.

First, we argue that $ \delta (s, t) \geq \dist_{G \setminus D} (s, t) $.
Observe that whenever we set the value of $ \delta_i (s, v, t) $ (for some $ 1 \leq i \leq \lceil \log{h} \rceil $ and some $ v \in R_i $) during the preprocessing or the deletion procedure, this value corresponds to the length of a path in a subgraph of $ G $.
Furthermore, if $ \pi_i^v (s, v) $ or $ \pi_i^v (v, t) $ contains a deleted node from $ D $, then the value $ \delta_i (s, v, t) $ is updated to the length of a path in $G \setminus D$.
It is not hard to verify now that $ \delta (s, t) \geq \dist_{G \setminus D} (s, t) $.

We now argue that $ \delta (s, t) \leq \dist_{G \setminus D} (s, t) $.
Let $ \pi $ denote the designated shortest path from $ s $ to $ t $ in $ G \setminus D $.
Consider first the case that $ \pi $ consists of at most $ h $ edges (where $ h =  \sqrt{n / |D|} $, as set in the algorithm).
Let $ 1 \leq i \leq \lfloor \log{h} \rfloor $ be the index such that the number of edges of $ \pi $ is between $ 2^{i-1} $ and $ 2^i $.
By \Cref{lem:centers hit paths}, $ \pi $ contains a center of $ C_i $ with probability at least $ 1 - 1/n^c $.
As $ C_i \subseteq R_i $, $ \pi $ contains at least one node of $ R_i $.
Now let $ v $ be the node that, among all nodes of $ R_i $ contained in $ \pi $, has been visited first in the preprocessing.
Then all edges of $ \pi $ are contained in $ G_i^v \setminus D $.
If either $ s \in U_i^v $ or $ t \in U_i^v $, then the update algorithm has set $ \delta_i (s, v, t) = \dist_{H_i^v} (s, v) + \dist_{H_i^v} (v, t) $.
By \Cref{lem:sketch contains shortest paths with few hops}, we have $ \dist_{H_i^v} (s, v) = \dist_{G_i^v \setminus D} (s, v) $ and $ \dist_{H_i^v} (v, t) = \dist_{G_i^v \setminus D} (v, t) $.
It follows that
\begin{multline*}
\delta (s, t) \leq \delta_i (s, v, t) = \dist_{H_i^v} (s, v) + \dist_{H_i^v} (v, t) = \dist_{G_i^v \setminus D} (s, v) + \dist_{G_i^v \setminus D} (v, t) \leq \\ \dist_{G \setminus D} (s, v) + \dist_{G \setminus D} (v, t) = \dist_{G \setminus D} (s, t) \, .
\end{multline*}
If both $ s \notin U_i^v $ and $ t \notin U_i^v $, then we have set $ \delta_i (s, v, t) = \dist_G^{h_i} (s, v) + \dist_G^{h_i} (v, t) $ in the preprocessing.
As $ s \notin U_i^v $  and $ t \notin U_i^v $, the paths $ \pi_i^v (s, v) $ and $ \pi_i^v (v, t) $ both are contained in $ G \setminus D $ and thus
\begin{multline*}
\delta (s, t) \leq \delta_i (s, v, t) = \dist_G^{h_i} (s, v) + \dist_G^{h_i} (v, t) \leq \dist_{G \setminus D}^{h_i} (s, v) + \dist_{G \setminus D}^{h_i} (v, t) = \\ \dist_{G \setminus D} (s, v) + \dist_{G \setminus D} (v, t) = \dist_{G \setminus D} (s, v) \, .
\end{multline*}

Finally, consider the case that $ \pi $ consists of at least $ h $ edges.
Then, by another application of \Cref{lem:centers hit paths}, $ \pi $ contains a center $ v \in C_{\lceil \log{h} \rceil} $ and thus $ \delta (s, t) \leq \delta_{\lceil \log{h} \rceil} (s, t) \leq \dist_{G \setminus D} (s, v) + \dist_{G \setminus D} (v, t) = \dist_{G \setminus D} (s, t) $.
\end{proof}

\subsection{Running time}

In the following we analyze the time our algorithm needs for performing the preprocessing and for recomputing the all-pairs distances after a batch of deletions.
The running time guarantees hold with high probability as they depend on the size of the randomly chosen sets in the preprocessing.

\begin{claim}\label{lem:number of centers}
With probability at least $ 1 - 1 / n^c $, the size of $ C_i $ is $ O ((n \log{n}) / h_i) $ for all $ 1 \leq i \leq \lceil \log{n} \rceil $.
\end{claim}

The claim follows in a straightforward way from \Cref{lem:random hitting set} and its proof is thus omitted.
In the remainder of this section we omit the $ c $ in the $ O $-notation for readability.

\begin{claim}
The preprocessing takes time $ O (n^3 \log^2{n}) $ with probability at least $ 1 - 1 / n^c $.
\end{claim}

\begin{proof}
We analyze the running time of each iteration $ i $, where $ 1 \leq i \leq \lceil \log{n} \rceil $.
We obtain the set $ C_i $ of sampled nodes in time $ O (n) $.
Running $ h_i $ iterations of Bellman-Ford on a graph with at most $ n $ nodes takes time $ O (h_i n^2) $.
We perform one such computation for every node in $ R_i $ and $ |R_i| \leq 2 |C_i| = O ( (n \log{n}) / h_i) $ (by \Cref{lem:number of centers}).
Thus, the total time spent for the Bellman-Ford computations in iteration $ i $ is $ O (|R_i| n^2 h_i)) = O (n^3 \log{n}) $.
For updating the counters we iterate over all nodes in the $ \leq h_i $ shortest paths in total time $ O (|R_i| n h_i) = O (n^3) $.
For constructing the list $ L_i (s, t) $ for each pair of nodes $ s $ and $ t $ we sort at most $ n $ nodes and thus constructing all the lists takes time $ O (n^3 \log{n}) $.
It follows that the running time in each iteration is dominated by the term $ O (n^3 \log{n}) $.
As there are $ O (\log{n}) $ iterations, the total time spent in the preprocessing is $ O (n^3 \log^2{n}) $.
\end{proof}

To bound the time for processing a batch of deletions we first show that the special order in which we have visited the nodes in the preprocessing stage guarantees that only few nodes are affected by each deletion.
The fewer nodes are affected, the faster our algorithm runs.

\begin{claim}\label{lem:bound on congestion}
For each $ 1 \leq i \leq \lceil \log{n} \rceil $, every node is contained in at most $ O (h_i n \log{n}) $ of the stored shortest $ \leq h_i $ hop paths after the preprocessing, i.e., $ | \{ (x, v) \in V \times R_i \mid u \in \pi_i^v (x, v) \text{ or } u \in \pi_i^v (v, x) \} | = O (h_i n \log{n}) $.
\end{claim}

\begin{proof}
Observe that we can prove the claim by showing that for every node $ u $ the counter $ c_i (u) $ is $ O (h_i n \log{n}) $ at any time.
In the preprocessing we alternate between (a) visiting the center node with maximum counter and (b) visiting the non-center node with maximum counter (if such a node exists).
After visiting a node~$ v $, the sum of the counters increases by at most $ 2 h_i n $ as for every node $ x $ the number of nodes on $ \pi_i^v (x, v) $ and $ \pi_i^v (x, v) $ (except for $ v $ itself) is at most $ h_i $, respectively.
A visited node is removed from the graph and thus not visited again in iteration $ i $.

Consider the following two processes for distributing stones on piles.
In process~$ 1 $ we have $ k_1 = | C_i | $ piles, one for each center, and in process~$ 2 $ we have $ k_2 = | V \setminus C_i | $ piles, one for each non-center node.
Every time our algorithm visits a node $ v $ we do the following:
If $ v $ is a center node, we remove the corresponding pile (and the nodes it contains) from process~$ 1 $.
Similarly, if $ v $ is a non-center node, we remove the corresponding pile (and the nodes it contains) from process~$ 2 $.
After visiting $ v $, the algorithm increases the counters of certain nodes and we add the corresponding number of stones to the piles in processes $ 1 $ and~$ 2 $.

Observe that whenever we remove a pile from process~$ 1 $ this pile always carries the maximum number of stones and between any two removals of piles we have added at most $ 4 h_i n $ stones to the piles.
Therefore, the total number of stones on any pile of process~$ 1 $ (and thus the maximum counter of any center) is $ O (h_i n \log{k_1}) = O (h_i n \log{n}) $ by \Cref{lem:distributing stones on piles}.
By the same reasoning the maximum counter of any non-center node is $ O (h_i n \log{n}) $ as well.
\end{proof}

\begin{claim}\label{lem:time per deletion}
The time for processing a single batch of deletions given by a set $ D $ is $ O (n^{2.5} \sqrt{|D|} \log{n}) $.
\end{claim}

\begin{proof}
For every $ 1 \leq i \leq \lfloor \log{h} \rfloor $ and every $ v \in R_i $, we compute the set of affected nodes $ U_i^v $ by iterating over all shortest $ \leq h_i $ hop paths to and from $ v $ determined in the preprocessing.
This takes time $ O (h_i n) $ for fixed $ v $ and $ i $ and thus time $ O (\sum_{i=1}^{\lfloor \log{h} \rfloor} |R_i| h_i n) = O (n^2 \log^2{n}) $ in total.

Constructing all sketch graphs and then running Dijkstra's algorithm on them takes time $ O (\sum_{i=1}^{\lfloor \log{h} \rfloor} \sum_{v \in R_i} (|E_i^v| + n \log{n}) ) $.
For every $ 1 \leq i \leq \lfloor \log{h} \rfloor $ we bound the size of the sketch graph $ H_i^v $ by $ |E_i^v| \leq | U_i^v | n + 2n $ as each node in $ U_i^v $ contributes at most $ n $ edges to all its neighbors and each other nodes contributes at most $ 2 $ edges.
By \Cref{lem:bound on congestion}, each deleted node of $ D $ is contained in $ O (h_i n \log{n}) $ of the $ \leq h_i $ hop shortest paths determined in the preprocessing.
Therefore the total number of nodes affected by the deletions is $ \sum_{v \in R_i} |U_i^v| = O (|D| h_i n \log{n}) $ and thus $ \sum_{v \in R_i} |E_i^v| = O (|D| h_i n^2 \log{n} + |R_i| n) $.
Dijkstra's algorithm in all sketch graphs of layers $ 1 $ to $ \lfloor \log{h} \rfloor $ then takes total time
\begin{multline*}
O \left( \sum_{i=1}^{\lfloor \log{h} \rfloor} \sum_{v \in R_i} (|E_i^v| + n \log{n}) \right) \leq
O \left( \sum_{i=1}^{\lfloor \log{h} \rfloor} (|D| h_i n^2 \log{n} + |R_i| n \log{n}) \right) \leq \\
O \left( \sum_{i=1}^{\lfloor \log{h} \rfloor} (|D| 2^i n^2 \log{n} + n^2 \log{n}) \right) \leq
O (|D| h n^2 \log{n} + n^2 \log^2{n}) \leq O (n^{2.5} \sqrt{|D|} \log{n}) \, .
\end{multline*}

To compute the minimum $ \delta (s, t) = \min_{v \in R_i} \delta_i (s, v, t) $ we do the following.
We first compute the minimum over all values for which $ \delta_i (s, v, t) $ has changed since the preprocessing.
Then, we find the first value of the sorted list $ L_i (s, t) $ for which the value $ \delta_i (s, v, t) $ computed in the preprocessing has \emph{not} changed.
Clearly, the minimum of both values gives $ \min_{v \in R_i} \delta_i (s, v, t) $ and both steps takes time proportional to the number of changed values, which is $ O (|D| h_i n^2 \log{n}) $.

Finally, we bound the time spent on running Dijkstra's algorithm for every node $ v \in C_{\lceil \log{h} \rceil} $ and computing $ \delta_{\lceil \log{h} \rceil} (s, t) $ for every pair of nodes $ s $ and~$ t $.
By the sampling procedure, the size of $ C_{\lceil \log{h} \rceil} $ is $ O ((n \log{n}) / h) $ and thus both of these steps take time $ O ((n^3 \log{n}) / h) = O (n^{2.5} \sqrt{|D|} \log{n}) $.
\end{proof}

\section{Extensions and Additional Results}
\subsection{Negative edge weights}\label{sec:negative_weights}
We extend the reduction of \Cref{lem:decremental_to_fully_dynamic} such that the graph may have negative weights, but no negative cycles, while the decremental algorithm only needs to work with non-negative edge weights.
A potential function $ p $ is a function that assigns a value $ p (v) $ to every node $ v $ such that, for every edge $ (u, v) \in E $, $ w (u, v) + p (u) - p (v) \geq 0 $.
Such a potential exists if and only if the graph contains no negative cycle.
We call $ w_p (u, v) = w (u, v) + p (u) - p (v) $ the modified weight of the edge $ (u, v) $ under the potential function~$ p $.
For any valid potential function the modified edge weights are obviously non-negative and it is well-known that the shortest paths under the original edge weights are the same as under the modified edge weights.
This is known as the ``Johnson transformation''~\cite{Johnson77}.

The reduction now has the following additional steps.
In the preprocessing we construct a graph $ G_q $ that contains an additional node $ q $ and an edge $ (q, v) $ of weight $ 0 $ to every other node $ v $.
On this graph we run the Bellman-Ford algorithm in time $ O (m n) = O (n^3) $ which either detects a negative cycle or computes $ \dist_{G_q} (q, v) $ for every node $ v $.
It is well-known that $ p (v) = \dist_{G_q} (q, v) $ is a valid potential function.
Observe that this potential function remains valid even when we delete edges or nodes from the graph.
Thus, the update procedure of our decremental algorithm computes the shortest paths of the graph correctly.
After running the updates of the decremental algorithm, we undo the potential transformation.
In the fully dynamic algorithm, we then deal with the up to $ 2 \Delta $ inserted nodes by running $ 2 \Delta $ iterations of the Floyd-Warshall algorithm, which by default can deal with negative edge weights.

\subsection{Unweighted graphs}\label{appendix:unweighted}

In unweighted graphs the shortest $ \leq h $ hop paths are identical to the shortest paths with at most $ h $ edges.
Thus, in the preprocessing, we can determine the shortest $ \leq h_i $ hop paths by performing breadth-first search up to depth $ h_i $ in time $ O (n^2) $.
The total time spent for layer $ i $ in the preprocessing is therefore $ O (n^3 / h_i) $.

We slightly alter the reduction  of \Cref{lem:decremental_to_fully_dynamic} to obtain the fully dynamic algorithm by using different value of $ \Delta $ for each layer $ i $.
Specifically, we use layer $ i $ of the decremental algorithm for $ 2 \Delta_i = 2 \sqrt{n} / (h_i \sqrt{\log{n}}) $ updates before we rebuild it.
At every update, besides running each layer of the decremental algorithm we reconstruct shortest paths consisting of more than $ h = \sqrt{n} $ edges in time $ O ((n^3 \log{n}) / h) $ and run the Floyd-Warshall algorithm for at most $ 2 \Delta = 2 \sqrt{n} $ iterations in time $ O (\Delta n^2) $ to handle the latest $ 2 \Delta $ insertions (note that $ \Delta \geq \Delta_i $ for all $ 1 \leq i \leq \lfloor \log{h} \rfloor $).
Thus, the total update time is
\begin{equation*}
O \left( \sum_{1 \leq i \leq \lfloor \log{h} \rfloor} \left( \frac{n^3}{h_i \Delta_i} + \Delta_i h_i n^2 \log{n} + \Delta n^2 \right) + \frac{n^3 \log{n}}{h} \right) = O (n^{2 + \sfrac{1}{2}} \log^{\sfrac{3}{2}}{n}) \, .
\end{equation*}

\subsection{Deterministic algorithm}\label{appendix:deterministic}

In the following we sketch a deterministic fully dynamic APSP algorithm with a worst-case update time of $ O (n^{2 + \sfrac{3}{4}} \log^{\sfrac{3}{4}}{n}) $.
We again design an algorithm that, after some preprocessing, can handle a single batch of up to $ \Delta $ deletions and turn this into a fully dynamic algorithm using the reduction of \Cref{lem:decremental_to_fully_dynamic}.
The parameters of our algorithm are $ \Delta $ and $ h $ and we will explain how to set them optimally in the running time analysis.

In the preprocessing we prepare a data structure from which the shortest paths with at most~$ h $ edges can be reconstructed after any batch of at most $ \Delta $ deletions.
We again keep a congestion counter for each node and in each round visit the node with the maximum counter until \emph{all} nodes have been visited.
When we visit a node~$ v $, we compute the shortest $ \leq h $ hop paths to and from $ v $, denoted by $ \pi^v (x, v) $ and $ \pi^v (v, x) $, by running $ h $ iterations of the Bellman-Ford algorithm in the graph $ G^v $ from which all edges incident to previously visited nodes have been removed.
We then update the congestion counters according to the number of appearances of each node in these paths.
By \Cref{lem:distributing stones on piles}, this order of visiting nodes guarantees that after the preprocessing, for every node $ u $, there are at most $ O (h n \log{n}) $ pairs $ (x, v) $ or $ (v, x) $ such that $ v $ is contained in $ \pi^v (x, v) $ or $ \pi^v (v, x) $, respectively.
Again, the idea is that if we later on delete $ u $ we only have to repair the shortest paths for such pairs.

When we delete a set of nodes $ D $, we first determine, for every node $ v $, the set of affected nodes $ U^v $ consisting of all nodes $ x $ such that $ \pi^v (s, v) $ or $ \pi^v (v, t) $ was destroyed by one of the deletions.
We then build a sketch graph $ H^v $ as follows:
For every affected node $ x \in U^v $ we add edges to all neighbors of $ x $ to $ H^v $ and for every non-affected node $ x \notin U^v $ we add the edge to the successor on $ \pi^v (x, v) $ as well as the edge to the predecessor on $ \pi^v (x, v) $ to $ H^v $.
We then run Dijkstra's algorithm in $ H^v $ to recompute the shortest paths to and from $ v $ with at most $ h $ edges in $ G^v \setminus D $.
Finally, for every pair of nodes $ s $ and $ t $ we find the ``middle node'' $ v $ connecting $ s $ and $ t $ in the shortest way, either by using paths from the preprocessing stage or new paths computed during the update procedure.
This gives us all shortest paths that have at most $ h $ edges.
Using this information we can then compute all shortest paths with more edges deterministically by finding a set of centers that ``hits'' all the shortest paths with at most $ h $ edges using a greedy heuristic.
To the best of our knowledge this method was first described in~\cite{Zwick02}.

\begin{lemma}[\cite{Zwick02}]\label{lem:from_small_hop_to_full_distance}
For every (weighted directed) graph $ G $ and every $ h \geq 1 $, given all shortest paths in $ G $ that have at most $ h $ hops, one can compute $ \dist_G (s, t) $ for all pairs of nodes $ s $ and $ t $ (and the corresponding shortest paths) in time $ O (h n^2 + n^3 \log{n} / h) $.
\end{lemma}

The correctness and the running time follow the same arguments as for the randomized algorithm in \Cref{sec:randomized algorithm}.
However, the deterministic algorithm is slower because we do not know how to extend it to the multilayer approach of the randomized algorithm.
The preprocessing time is dominated by the term $ O (n^3 h) $, which is the time needed for running $ h $ iterations Bellman-Ford algorithm for all nodes.
As in the randomized algorithm, the running time for processing a single batch of deletions is dominated by the time needed to run Dijkstra's algorithm in each sketch graph where the total size of all sketch graphs is $ O (\Delta h n^2 \log{n}) $.
Thus, running Dijkstra's algorithm in all sketch graphs takes total time $ O (\Delta h n^2 \log{n} + n^2 \log{n}) = O (\Delta h n^2 \log{n}) $.
Overall, the update time of this deterministic fully dynamic algorithm is $ O (n^3 h / \Delta + \Delta h n^2 \log{n} + h n^2 + n^3 \log{n} / h + \Delta n^2) $.
By setting $ \Delta = n^{\sfrac{1}{2}} / \log^{\sfrac{1}{2}}{n} $ and $ h = n^{\sfrac{1}{4}} / \log^{\sfrac{1}{4}}{n} $, this is $ O (n^{2 + \sfrac{3}{4}} \log^{\sfrac{3}{4}}{n}) $.

\subsection{Returning shortest paths}\label{sec:returning shortest paths}

% http://cstheory.stackexchange.com/questions/17575/shortest-paths-perturbation

The algorithm of \Cref{thm:worst_case_randomized} we have formulated above only returns the distance matrix for all pairs of nodes.
Our decremental algorithm be extended in a straightforward way to return a shortest path matrix that contains, for every pair of nodes $ s $ and $ t $ the first edge $ (s, x) $ on a shortest path from $ s $ to $ t $ by additionally storing with every intermediate distance estimate the first edge on the path yielding the corresponding value.
The Floyd-Warshall algorithm can also compute the shortest path matrix along with the distances.
This matrix provides sufficient information to compute a shortest path between a source and a target in time proportional to the number of edges of this path.
Without additional effort however, we can only show that such an algorithm is correct with high probability against an oblivious adversary who chooses its sequence of updates and queries in advance.
In the following we sketch how to modify our algorithm such that it is correct with high probability against an adaptive online adversary that may adapt its sequence of updates according to the shortest path matrix returned by the algorithm.
Intuitively, this means that we have to avoid that the adversary chooses a clever sequence of updates and queries such that it can identify the random centers picked by our algorithm and alter the graph such that the properties we gained from the random choice of centers will not hold anymore.
Note that this is not an issue for algorithms that merely compute the distance matrix since the exact distances are uniquely defined in every graph and thus the adversary cannot learn anything about our algorithm (and its random choices) by observing its outputs.
If we allow path queries, where the adversary may pick any pair of nodes $ s $ and $ t $ and ask for a shortest path from $ s $ to $ t $, then we do not have this uniqueness property because there might be several paths from $ s $ to $ t $ of minimum weight.

We counter this problem by modifying our algorithm such that it computes lexicographic shortest paths at the cost of an additional factor of $ \log{n} $ in the update time.\footnote{An intuitive alternative for enforcing that shortest paths are unique is to add random perturbation to the edge weights. However we do not know how to use this scheme in our dynamic setting. While one can guarantee that a random perturbation makes shortest paths unique in the initial graph, it is not clear how to obtain this property in all versions of the graph.}
First, we assume without loss of generality that all shortest paths of the graph have the same number of edges.
This can be ensured by adding a sufficiently small penalty to every edge weight.
Now assume an arbitrary but fixed order on the nodes and identify a path with its sequence of nodes.
We say that a path $ \pi_1 $ is lexicographically smaller than a path $ \pi_2 $ if either $ \pi_1 $ is a prefix of $ \pi_2 $ or $ \pi_1 = \pi \circ v_1 \circ \pi_1'' $ and $ \pi_2 = \pi \circ v_2 \circ \pi_2'' $ where $ v_1 < v_2 $.
The latter condition means that at the first position where $ \pi_1 $ and $ \pi_2 $ differ, the node of $ \pi_1 $ is a smaller than the one of $ \pi_2 $.
The lexicographic shortest path from $ s $ to $ t $ is the path that among all shortest paths from $ s $ to $ t $ is smallest according to the lexicographic order.
To compare paths lexicographically we use the \emph{minimum-index tree structure} (short: MITS) of Cabello, Chambers, and Erickson~\cite{CabelloCE13}, which can be implemented using Euler-tour trees~\cite{HenzingerK99, Tarjan97} or self-adjusting top trees~\cite{TarjanW05}.
This minimum index tree structure allows adding and removing edges to and from a tree with given root $ s $.
Given two nodes $ u $ and $ v $ as its input it can answer the following query: is the path from $ s $ to $ u $ in the tree lexicographically smaller than the path from $ s $ to $ v $?
All operations of this data structure take time $ O (\log{n}) $.
Cabello, Chambers, and Erickson~\cite{CabelloCE13} explain how use the MITS to adapt Dijkstra's algorithm for computing lexicographic shortest paths at the cost of an additional factor of $ O (\log{n}) $ in the running time.
In a similar way the Bellman-Ford algorithm can be modified to compute the lexicographic $ \leq h $ hop shortest paths.

We now explain how to modify our algorithm to compute, for every pair of nodes $ s $ and $ t $, the first edge $ (s, x) $ on the lexicographic shortest path from $ s $ to $ t $ in $ G \setminus D $.
In the preprocessing algorithm the only modification is to use the modified Bellman-Ford algorithm to compute the lexicographic $ \leq h_i $ hop shortest paths.
In the update procedure we first run the algorithm for computing the distance matrix completely so that we know the value of $ \dist_{G \setminus D} (s, t) $ for all pairs of nodes $ s $ and~$ t $.
We then initialize an MITS for every node $ s $ and add to it the following paths computed by our algorithm.
\begin{enumerate}
\item For every $ 1 \leq i \leq \lfloor \log{h} \rfloor $ and every $ v \in R_i $ we add the shortest $ \leq h_i $ hop path $ \pi_i^v (s, v) $ from the preprocessing to the MITS if the weight of $ \pi_i^v (s, v) $ is equal to $ \dist_{G \setminus D} (s, v) $.
\item For every $ 1 \leq i \leq \lfloor \log{h} \rfloor $ and every $ v \in R_i $ such that $ s \in U_i^v $, we add the lexicographic shortest path $ \pi $ from $ s $ to $ v $ in $ H_i^v $ to the MITS if $ \pi $ has at most $ h_i $ edges and the weight of $ \pi $ is equal to $ \dist_{G \setminus D} (s, v) $.
\item For every $ v \in R_{\lceil \log{h} \rceil} $ we add the lexicographic shortest path from $ s $ to $ v $ in $ G \setminus D $ to the MITS.
\end{enumerate}
Using similar arguments as in \Cref{sec:correctness} it follows that (1) we only add lexicographic shortest paths to the MITS leading to a tree structure and (2) for every lexicographic shortest path $ \pi $ from $ s $ to some node $ t $ there is some node $ v $ on $ \pi $ such that the subpath from $ s $ to $ v $ is contained in the MITS of $ s $.
The running time for these insertions of paths to the MITS's can be bounded as follows:
As, for each $ 1 \leq i \leq \lfloor \log{h} \rfloor $, $ R_i = O ((n \log{n}) / h_i) $ and each inserted path has at most $ h_i $ edges, the total time for the insertions in step 1 is $ O (n^2 \log^3{n}) $.
The total time for step 2 can be bounded by recalling that the total number of affected nodes is $ \sum_{v \in R_i} | U_i^v | = O (\Delta h_i n \log{n}) $, yielding a total time of $ O (\Delta h_i n^2 \log^2{n}) $.
Since $ | R_{\lceil \log{h} \rceil} | = O ((n \log{n}) / h) $ the total time for step 3 is $ O ((n^3 \log^2{n}) / h) $.

To obtain the lexicographic shortest paths of $ G \setminus D $ for every pair of nodes $ s $ and $ t $ (or rather the first edge of this path) we do the following:
We let $ X (s, t) $ be the set of nodes $ v $ satisfying one of the following three conditions:
\begin{enumerate}
\item $ v \in R_i $ for some $ 1 \leq i \leq \lfloor \log{h} \rfloor $ and $ \pi_i^v (s, v) + \pi_i^v (v, t) = \dist_G (s, t) $ ,
\item $ v \in R_i $ such that $ s \in U_i^v $ or $ t \in U_i^v $ for some $ 1 \leq i \leq \lfloor \log{h} \rfloor $ and $ \dist_{H_i^v} (s, v) + \dist_{H_i^v} (v, t) = \dist_G (s, t) $, or
\item $ v \in R_{\lceil \log{h} \rceil} $ and $ \dist_{G \setminus D} (s, v) + \dist_{G \setminus D} (v, t) = \dist_G (s, t) $.
\end{enumerate}
Among the nodes in $ X (s, t) $ we now want to find the node $ v $ whose lexicographic shortest path from $ x $ to~$ v $ is smallest in $ G \setminus D $ and output the first edge on this path.
This can be done in time $ O (| X (s, t) | \log{n}) $ using the MITS of $ s $ that we have prepared above.
By our arguments from above, the total size of all these sets is bounded by $ O (|D| h n^2 \log{n} + (n^3 \log{n}) / h) $.
Using $ h = \sqrt{n / |D|} $ we thus obtain a running time of $ O (n^{2.5} \sqrt{|D|} \log{n}) $.

\section{Conclusions}

In this paper we considered the fully dynamic APSP problem with a worst-case update time guarantee of $ O (n^{2+\sfrac{2}{3}} \log^{\sfrac{4}{3}}{n}) $.
Our algorithm is simple and independent of any other sophisticated algorithms. 
Our current knowledge of lower bounds for this problem seems quite rudimentary. 
A natural barrier for the current approaches seems to be $\Omega(n^{2+\sfrac{1}{2}})$. 
One reason for this barrier is that the only way we know to deal with insertions is to use the naive approach, in which for every insertion (since the last time the data structure was reconstructed) in every update we compute a SSSP tree and recompute all pair-wise distances in these trees.
This naive approach sets a barrier of $\Omega(n^{2+\sfrac{1}{2}})$.
A natural question is on the existence of fully dynamic APSP algorithm that meet this barrier or prove impossibility results.

For unweighted graphs, our upper bound indeed meets this barrier. Weighted graphs seem to be inherently harder: for example, extending the algebraic techniques of Sankowski \cite{Sankowski05} to weighted graphs is an open question.  Our techniques for weighted graphs incur a cost related to computing single source $h$-hop shortest paths: the best known is time $\tilde O(n^2 h)$ for weighted graphs, and time $\tilde O(n^2)$ for unweighted graphs.
If $h$-hop shortest paths could be solved in time $\tilde O(n^2)$ for weighted graphs, then our techniques would immediately provide  an improved results that would meet the natural $\Omega(n^{2+\sfrac{1}{2}})$ barrier.

We believe that it would be interesting to also explore the potential opposite connection: could hardness of $h$-hop shortest paths in weighted graphs imply lower bounds for dynamic shortest paths in weighted graphs? 
Another interesting direction is to derandomize our algorithm or prove an existential gap between randomized and deterministic algorithms.

\printbibliography[heading=bibintoc] % Make bibliography show up in table of contents

\end{document}